\RequirePackage{fix-cm}
\documentclass[smallextended]{svjour3}       
\smartqed  
\usepackage{graphicx}

\usepackage{mathtools}

\usepackage[utf8]{inputenc}
\usepackage{amsfonts,amssymb,amsmath,enumerate}

\usepackage{tabularx,hhline,rotating,xspace}

\usepackage{stmaryrd}

\usepackage{tikz}
\usetikzlibrary{calc}
\usetikzlibrary{positioning}

\usepackage{multicol}
 \usepackage[hidelinks]{hyperref}

\renewcommand{\setminus}{\mysetminus}

\newenvironment{aw}{\noindent\color{magenta} AW : }{}

\newcommand{\mysetminusD}{\raisebox{.8pt}{\hbox{\tikz{\draw[line width=0.6pt,line cap=round] (3.5pt,0pt) -- (0,5.2pt);}}}}
\newcommand{\mysetminusT}{\mysetminusD}
\newcommand{\mysetminusS}{\raisebox{.5pt}{\hbox{\tikz{\draw[line width=0.45pt,line cap=round] (2.2pt,0) -- (0,3.8pt);}}}}
\newcommand{\mysetminusSS}{\raisebox{.35pt}{\hbox{\tikz{\draw[line width=0.4pt,line cap=round] (1.5pt,0) -- (0,2.8pt);}}}}

\newcommand{\mysetminus}{\mathbin{\mathchoice{\mysetminusD}{\mysetminusT}{\mysetminusS}{\mysetminusSS}}}

\newcommand{\nilRes}[1]{\gamma_{\omega}(#1)}
\newcommand{\Nleq}{\mathrel{\trianglelefteqslant}}
\newcommand{\Nle}{\mathrel{\lhd}}

\newcommand{\Ncl}[1]{\left<\!\left< \mathinner{#1} \right>\!\right>}

\newcommand{\set}[2]{\left\{\, \mathinner{#1}\vphantom{#2}\: \left|\: \vphantom{#1}\mathinner{#2} \right.\,\right\}}
\newcommand{\oneset}[1]{\left\{\, \mathinner{#1} \,\right\}}
\newcommand{\interval}[2]{[ \mathinner{#1}..\mathinner{#2}] }

\newcommand{\smallset}[1]{\left\{\mathinner{#1}\right\}}

\newcommand{\abs}[1]{\left|\mathinner{#1}\right|}

\newcommand{\ceil}[1]{\left\lceil\mathinner{#1} \right\rceil}

\newcommand{\gen}[1]{\left< \mathinner{#1} \right>}
\newcommand{\genr}[2]{\left< \, \mathinner{#1}\; \middle|\;\mathinner{#2} \, \right>}

\newcommand{\Fit}{\operatorname{Fit}}
\newcommand{\FitL}{\operatorname{FitLen}}

\newcommand{\mcomm}[3]{\left[{#1},\kern.1em_{#2}\,\kern.1em {#3} \right]}
\newcommand{\smcomm}[2]{\left[_{#1}\,\kern.1em {#2} \right]}

\newcommand{\N}{\ensuremath{\mathbb{N}}}
\newcommand{\Z}{\ensuremath{\mathbb{Z}}}

\newcommand{\NP}{\ensuremath{\mathsf{NP}}\xspace} %
\newcommand{\coNP}{\ensuremath{\mathsf{coNP}}\xspace}
\newcommand{\TC}{\ensuremath{\mathsf{TC}^0}\xspace}

\newcommand{\CC}{\ensuremath{\mathsf{CC}^0}\xspace}

\renewcommand{\P}{\ensuremath{\mathsf{P}}\xspace}

\renewcommand{\phi}{\varphi}

\newcommand{\Oh}{\mathcal{O}}

\newcommand{\cU}{\mathcal{U}}

\newcommand{\cX}{\mathcal{X}}

\newcommand{\SAT}{\textsc{3SAT}\xspace}

\newcommand{\KColoring}[1]{\ensuremath{#1}\textsc{-Coloring}\xspace}

\newcommand{\sse}{\subseteq}

\newcommand\ie{i.e., }

\newcommand\eg{e.g.\xspace}

\newcommand{\m}[1]{\uppercase{#1}}
\newcommand{\po}[1]{{\mathbf{#1}}}              
\renewcommand{\o}[1]{\overline {#1}}            
\newcommand{\el}{\ell}                          
\renewcommand{\leq}{\leqslant}
\renewcommand{\geq}{\geqslant}
\renewcommand{\vec}{\o}

\newcommand{\polsat}[1]{\ensuremath{\operatorname{\textsc{{PolSat}}}
		\ifthenelse{\equal{#1}{}}{}{\!\left( {\m #1} \right)}}}
\newcommand{\poleqv}[1]{\ensuremath{\operatorname{\textsc{{PolEqv}}}
		\ifthenelse{\equal{#1}{}}{}{\!\left( {\m #1} \right)}}}

\newcommand{\card}[1]{\left| #1 \right|}

\newcommand{\npc}{\NP-complete\xspace}
\newcommand{\conpc}{\coNP-complete\xspace}
\newcommand{\np}{\NP}
\newcommand{\conp}{\coNP}
\newcommand{\ptime}{\P}

\newcounter{senumi}[section]
\newcounter{senumip}[section]
\newcounter{temp}[section]

\def\thesenumi{\thesection.\arabic{senumip}}
\def\p@senumip\thesenumip{\thesenumi}
\newenvironment{senumerate}%
    {\begin{list}%
        {(\thesenumi)}%
        {\usecounter{senumip}}
        \setcounter{senumip}{\value{temp}}
    }%
    {\setcounter{temp}{\value{senumip}}
     \end{list}}
\newcommand{\sref}[1]{(\thesection.\ref{#1})}

\newcommand{\myqed}{\hfill$\Box$}

\newcommand{\polyD}{\po{p}}

\newcommand{\polyA}{\po{s}}
\newcommand{\polyB}{\po{r}}
\newcommand{\polySAT}{\po{s}}
\newcommand{\polyAND}{\po{r}}

\begin{document}

\title{Equation satisfiability in solvable groups\thanks{The research of the first three authors was partially supported by Polish NCN Grant no 2014/14/A/ST6/00138.
The fourth author has been funded by DFG project DI 435/7-1.}
}


\author{Pawe\l{}~Idziak \and
	Piotr~Kawa\l{}ek \and
	Jacek~Krzaczkowski \and
	Armin~Weiß
}


\institute{Pawe\l{} Idziak \at
		Jagiellonian University, Faculty of Mathematics and Computer Science, \\
        Theoretical Computer Science Department \\
		\email{pawel.idziak@uj.edu.pl}           
	\and
		Piotr Kawa\l{}ek \at
		Jagiellonian University, Faculty of Mathematics and Computer Science, \\
        Theoretical Computer Science Department \\
		\email{piotr.kawalek@doctoral.uj.edu.pl}           
	\and	
		Jacek Krzaczkowski \at
        Maria Curie-Sklodowska University, Faculty of Mathematics, Physics and Computer Science,\\ Institute of Computer Science\\
        \email{krzacz@poczta.umcs.lublin.pl}
	\and
		Armin Weiß   \at
		Universität Stuttgart, Institut für Formale Methoden der Informatik (FMI)\\
		Universitätsstr. 38,
		70569 Stuttgart,
		Germany\\
		\email{armin.weiss@fmi.uni-stuttgart.de}
}

\date{Received: date / Accepted: date}

\maketitle

\begin{abstract}
	The study of the complexity of the equation satisfiability problem in finite groups had been initiated by Goldmann and Russell in \cite{GoldmannR02}
	where they showed that this problem is in \ptime
	for nilpotent groups  while it is \npc for non-solvable groups.
	Since then, several results have appeared showing that the problem can be solved in polynomial time in certain solvable groups $\m G$ having a nilpotent normal subgroup $H$ with nilpotent factor $\m G/H$.
	This paper shows that such normal subgroup must exist in each finite group with equation satisfiability solvable in polynomial time, unless the Exponential Time Hypothesis fails.

\end{abstract}


\section{Introduction}

The study of equations over algebraic structures has a long history in mathematics. Some of the first explicit decidability results in group theory are due to Makanin \cite{mak77}, who showed that equations over free groups are decidable. Subsequently several other decidability and undecidability results as well as complexity results on equations over infinite groups emerged (see \cite{DiekertE17icalpshort,GarretaMO20,LohSen06,Romankov79} for a random selection).
Also the famous 10th Hilbert problem on Diophantine equations,
that asks whether an equation of two polynomials over the ring of integers has a solution,
was shown to be undecidable \cite{Matijasevic}.

One can treat polynomials over a ring $\m R$ to be terms over $\m R$ with some variables already evaluated by elements of $\m R$.
The same can be done with groups to define polynomials over a group $\m G$.
Now the problem $\polsat{G}$ takes as input an equation of the form
$\po t(x_1,\ldots,x_n) = \po s(x_1,\ldots,x_n)$
(or equivalently $\po t(x_1,\ldots,x_n) = 1$,
by replacing $\po t = \po s$ by $\po t \po s^{-1}$),
where $\po s(\o x)$ and $\po t(\o x)$ are polynomials over $\m G$,
and asks whether this equation has a solution in $\m G$.
Obviously working with terms $\po t, \po s$ rather than polynomials this problem trivializes by setting all the $x_i$'s to $1$.
Likewise $\poleqv{G}$ is the problem of deciding whether two polynomials
$\po t(\o x), \po s(\o x)$ are equal for all evaluations of the variables $\o x$ in $\m G$.

While for infinite groups $\m G$ the problems $\polsat{G}$ and $\poleqv{G}$ may be undecidable, they are solvable in exponential time in finite realms.
In fact, $\polsat{G}$ is in \np, whereas $\poleqv{G}$ is in \conp.
Actually the hardest possible groups that lead to \npc $\polsat{}$ and \conpc $\poleqv{}$
are all groups that are not solvable \cite{GoldmannR02,HorvathM07}.
On the other hand it is easy to see that both these problems can be solved in a linear time for all finite abelian groups.

Also in nilpotent groups both $\polsat{}$ and $\poleqv{}$ can be solved in polynomial time.
While the running time of the first such algorithm for $\polsat{}$, due to Goldmann and Russell \cite{GoldmannR02}, is bounded by a polynomial of very high degree (as this bound was obtained by a Ramsey-type argument), the first algorithm for $\poleqv{}$ (due to \cite{BurrisL04}) is much faster.
For polynomials of length $\el$ the running time for $\poleqv{G}$ is bounded by
$\Oh\left(\el^{k+1}\right)$, where $k \leq \log \card{G}$ is the nilpotency class of the group $\m G$.
Very recently two much faster algorithms for $\polsat{G}$ have been described.
One by \cite{Foldvari18} runs in $\Oh\left( \el^{\frac{1}{2} \card{G}^2 \log\card{G}}  \right)$ steps.
The other one, provided in \cite{KawalekK}, runs even faster for all but finitely many nilpotent groups,
i.e. in $\Oh\left( \el^{\card{G}^2+1}\right)$ steps.
The very same paper \cite{KawalekK} concludes this race by providing randomized algorithms for $\polsat{}$ and $\poleqv{}$ working in linear time for all nilpotent groups.

\medskip

However, the situation for solvable but non-nilpotent groups has been almost completely open.
Due to \cite{HorvathS06} we know that $\polsat{}$ and $\poleqv{}$ for the symmetric group $\m S_3$ (and some others) can be done in polynomial time.
More examples of such solvable but non-nilpotent groups are provided in \cite{Horvath15,FoldvariH19}.
Actually already in 2004 Burris and Lawrence \cite{BurrisL04} conjectured that $\poleqv{}$ for all solvable groups is in \ptime.
In 2011 Horv\'ath renewed this conjecture and extended it to $\polsat{}$ \cite{Horvath11}.
Actually these conjectures have been strongly supported also by recent results in \cite{FoldvariH19}, where many other examples of solvable non-nilpotent groups are shown to be tractable.

Up to recently, the smallest solvable non-nilpotent group with unknown complexity was the symmetric group $\m S_4$. One reason that prevented existing techniques for polynomial time algorithms to work for $\m S_4$ is that $\m S_4$ does not have a nilpotent normal subgroup with a nilpotent quotient.
Somewhat surprisingly, in \cite{IdziakKK20} the first three authors succeeded to show that neither $\polsat{S_4}$ nor $\poleqv{S_4}$ is in \ptime as long as the Existential Time Hypothesis holds.
Simultaneously, in \cite{Weiss20} the fourth author proved super-polynomial lower bounds on $\polsat{}$ and $\poleqv{}$ for a broad class of finite solvable groups~-- again unless ETH fails. Both the lower bounds in \cite{IdziakKK20} and \cite{Weiss20} depended on the so-called Fitting length, which is defined as the length $d$ of the shortest chain
\[
1 = G_0 \leq G_1 \leq \ldots \leq G_d = \m G
\]
of normal subgroups $G_i$ of $\m G$ with all the quotients $G_{i+1}/G_i$ being nilpotent.
Indeed, the lower bounds in \cite{Weiss20} apply to all finite solvable groups of Fitting length at least four and to certain groups of Fitting length three. However, this class of groups does not include $\m S_4$~-- although its Fitting length is three.

The present paper extends these results by showing super-polynomial lower bounds for the complexity of $\polsat{G}$ and $\poleqv{G}$~-- again depending on the Fitting length. It strongly indicates that the mentioned conjectures by Burris and Lawrence and by Horv\'ath fail by showing the following result.

\begin{theorem}
	\label{thm:mainIntro}
	If $\m G$ is a finite solvable group of Fitting length $d\geq 3$,
	then both $\polsat{G}$ and $\poleqv{G}$ require at least
	$2^{\Omega(\log^{d-1} \el)}$ steps unless ETH fails.
\end{theorem}

The paper \cite{BarringtonMMTT00} contains all necessary pieces to provide for $\polsat{G}$
an upper bound of the form $2^{\Oh(\log^{r} \el)}$ with $r\geq 1$ depending on $\m G$ whenever $\m G$ is a finite solvable group.
This upper bound relies on the \textit{AND}-weakness conjecture saying that each $\CC$ circuit for the $n$-input \textit{AND} function has at least $2^{n^\delta}$ gates. Thus, the \textit{AND}-weakness conjecture implies that the lower bounds in \autoref{thm:mainIntro} cannot be improved in an essential way.

\medskip

Finally, we note that allowing to use definable polynomials as additional basic operations to built the input terms $\po t, \po s$ we may substantially shorten the size of the input.
For example with the commutator $[x,y]=x^{-1}y^{-1}xy$ the expression
$[\dots[[x,y_1], y_2],\ldots,y_n]$ has linear size, while when presented in the pure group language it has exponential size.
In this new setting \polsat{} (and \poleqv{}) have been shown \cite{HorvathS11,Kompatscher19}
to be \npc (or \conpc, respectively) for all non-nilpotent groups.
Actually our proof of Theorem \ref{thm:mainIntro} shows this as well.

Moreover, the paper \cite{IdziakK18} shows (in a very broad context of an arbitrary algebra) that allowing such definable polynomials can be simulated by circuits over this algebra.

\section{Preliminaries}\label{sec:prelims}

\paragraph{Complexity and the Exponential Time Hypothesis.}
We use standard notation from complexity theory as can be found in any textbook on complexity,
\eg \cite{pap94}.

The Exponential Time Hypothesis (ETH) is the conjecture that there is some $\delta > 0$ such that every algorithm for $\SAT$ needs time $\Omega(2^{\delta n})$ in the worst case where $n$ is the number of variables of the given $\SAT$ instance.
By the Sparsification Lemma \cite[Thm.~1]{ImpagliazzoPZ01} this is equivalent to the existence of some $\epsilon > 0$ such that every algorithm for $\SAT$ needs time $\Omega(2^{\epsilon (m+n)})$ in the worst case where $m$ is the number of clauses of the given $\SAT$ instance (see also \cite[Thm.~14.4]{CyganFKLMPPS15}). In particular, under ETH there is no algorithm for \SAT running in time $2^{o(n+m)}$.

Another classical \npc problem is \KColoring{C} for $C\geq 3$.
Given an undirected graph $\Gamma = (V,E)$ the question is whether there is a valid $C$-coloring of $\Gamma$, i.e. a map $\chi:V \to \interval{1}{C}$ satisfying
$\chi(u) \neq \chi(v)$ whenever $\smallset{u,v}\in E$.
Moreover, by \cite[Thm.~14.6]{CyganFKLMPPS15}, \KColoring{3} cannot be solved in time $2^{o(\abs{V} + \abs{E})}$ unless ETH fails. Since  \KColoring{3} can be reduced to  \KColoring{C} for fixed $C \geq 3$ by introducing only a linear number of additional edges and a constant number of vertices, it follows for every $C\geq 3$ that also \KColoring{C} cannot be solved in time $2^{o(\abs{V} + \abs{E})}$ unless ETH fails.

\paragraph{Groups and Commutators}
Throughout, we only consider finite groups $\m G$.
We follow the notation of \cite{Robinson96book}.
For a groups $\m G$ and $\m H$ we write $H\leq G$ if $H$ is a subgroup of $\m G$,
or  $H<G$ if $H$ is a proper subgroup of $G$.
Similarly we write $H\Nleq G$ (or $H\Nle G$) if $H$ is a normal subgroup of $\m G$
(or a proper normal subgroup).
For a subset $X \sse G$ we write $\gen{X}$ for the subgroup generated by $X$,
and $\Ncl{X} = \genr{x^g}{x\in X, g\in G}$ for the normal subgroup generated by $X$.

We write $[x,y] = x^{-1}y^{-1}xy$ for the commutator and $x^y = y^{-1}xy$ for the conjugation. Moreover, we write $[x_1, \dots, x_n] = [[x_1, \dots, x_{n-1}],x_n]$ for $n\geq 3$.

We will be also using commutator of (normal) subgroups (or even subsets) $X,Y,X_1,\dots, X_k \sse G$
defined by $[X,Y] = \genr{[x,y]}{x \in X, y \in Y}$
and $[X_1,\dots, X_k] = [[X_1,\dots,X_{k-1}], X_k]$.
Note here that the commutator $[H,K]$ is a normal subgroup of $\m G$ whenever $H$ and $K$ are.
Finally, we put $\mcomm{x}{k}{y} = [x,\underbrace{y,\dots,y}_{k\text{ times}}]$
and $\mcomm{X}{k}{Y} = [X,\underbrace{Y,\dots,Y}_{k\text{ times}}]$.

We will also need the concept of a centralizer of a subset $X$ in $G$, which is defined
as $C_G(X) = \set{g \in G}{[g,h]=1 \text{ for all } h \in X}$.
If $N$ is a normal subgroup, then $C_G(N)$ is a normal subgroup as well.

Below we collect some basic facts about commutators of elements and subgroups.

\begin{senumerate}
\item
\label{lem:commutatorBasics}
For $g,x,y,z,x_1, \dots x_n,y_1, \dots y_n  \in G$
 and normal subgroups $K_1, K_2,M,N$ of a group $G$ and
we have
\begin{enumerate}[(i)]
\item
\label{basicCommId}
$[xy,z] = [x,z]^y[y,z]$ \text{ and }  $[x,yz] = [x,z][x,y]^z$.
\item\label{joindistributive}
$[K_1,K_2] = [K_2,K_1] \leq K_1 \cap K_2$  \text{ and }  $[K_1K_2,N] = [K_1,N][K_2,N]$.
\item
\label{iteratedCommId}
If	$x \equiv y \mod N$ and $g \in M$, then for all $k\in \N$ we have
\[\mcomm{x}{k}{g}  \equiv \mcomm{y}{k}{g} \mod \mcomm{N}{k}{M}.\]
	\item\label{iteratedCommIdright} If $g \in M$ and $x_i \equiv y_i \mod N$, then \[[g,x_1, \dots, x_n] \equiv [g,y_1, \dots, y_n]\mod [M,N].\]
	\item\label{centralComm} For all $f \in C_G(N)$,  $g \in G$, $h \in N$ and $k \in \N$ we have \[\mcomm{hf}{k}{g} = \mcomm{h}{k}{g}\mcomm{f}{k}{g}.\]
\end{enumerate}
\end{senumerate}

\begin{proof}
For (\ref{basicCommId}), see \cite[5.1.5]{Robinson96book}.
The first part of (\ref{joindistributive}) is clear from the definition,
while the second one follows immediately from (\ref{basicCommId}).
To see (\ref{iteratedCommId}) and (\ref{iteratedCommIdright}),
let $g \in M$, $x,y \in G$ and $h \in N$ with $hx=y$ to see that
	\begin{align*}
		[hx,g] &= [h,g]^x[x,g] \in [N,M] [x,g]\qquad \text{ and }\\
		[g,hx] &= [g,x][g,h]^x \in [g,x] [M,N].
	\end{align*}
Then our statements follow by induction.	
	
	Finally, for (\ref{centralComm}), let $f \in C_G(N) = \set{g \in G}{[f,h]=1 \text{ for all } h \in N}$ and $g \in G$, $h \in N$. Then we have
	\begin{align*}
		[hf,g] &= [h,g]^f[f,g] = [h,g] [f,g].
	\end{align*}
Since $ C_G(N)$ is a normal subgroup, also $  [f,g] \in C_G(N) $ so that we can then induct on $k$.
\myqed
\end{proof}

\newcommand{\omegaG}{\omega}

Since $G$ is finite, for all $x,y \in G$, there are $i<j$ such that $\mcomm{x}{i}{y}=\mcomm{x}{j}{y}$. Writing $k = j-i$, we get $\mcomm{x}{i}{y}=\mcomm{x}{i+k}{y}$ for all sufficiently large $i$'s. For each choice of $x$ and $y$ we might get a different value for $k$; yet, by taking a common multiple of all the $k$'s, we obtain some $\omegaG \in \N$ such that for all $x,y \in G$ and all $i\geq \omegaG$ we have $\mcomm{x}{i}{y}=\mcomm{x}{i+\omegaG}{y}$.

Since for normal subgroups $M,N$ of $\m G$ we have
\[
M \geq \mcomm{M}{1}{N} \geq \mcomm{M}{2}{N} \geq \ldots \geq \mcomm{M}{i}{N} \geq \mcomm{M}{i+1}{N}\geq \ldots,
\]
the finiteness of $\m G$ ensures us that there is some $k_0 \in \N$ such that $\mcomm{M}{k_0}{N} = \mcomm{M}{k}{N}$ for all $k \geq k_0$ and all normal subgroups $M,N$ of $\m G$.
We can assume that $\omegaG \geq k_0$.
It is clear that $\omegaG =\abs{G}!$ is large enough, but typically much smaller values suffice.
Thus, we have:
\begin{senumerate}
\item For $x,y \in G$, $M,N\Nleq G$ and $i,j\geq \omegaG$ we have
\begin{itemize}
	\item $\mcomm{x}{i}{y}=\mcomm{x}{i+\omegaG}{y}$,
	\item $\mcomm{M}{i}{N}=\mcomm{M}{j}{N}$.
\end{itemize}
\end{senumerate}
We fix $\omegaG = \omegaG(G)$ throughout. Be aware that it depends on the specific group $G$.

\paragraph{Nilpotency and Fitting series.}
The $k$-th term of the lower central series is $\gamma_k(G) = \mcomm{G}{k}{G}$. The \emph{nilpotent residual} of $G$ is defined as $\bigcap_{k\geq 0 }\gamma_k(G) = \gamma_\omegaG (G)$ where $\omegaG$ is as above (\ie $\nilRes{G} = \gamma_i(G)$ for every $i \geq \omegaG$). Recall that a finite group $G$ is nilpotent if and only if $\nilRes{G} = 1$.

The \emph{Fitting} subgroup $\Fit(G)$ is the union of all nilpotent normal subgroups. Let $G$ be a finite solvable group. It is well-known that $\Fit(G)$ itself is a nilpotent normal subgroup (see \eg \cite[Satz 4.2]{Huppert67}).
We will need the following characterization of the Fitting subgroup due to Baer
(see \cite[Satz L']{Baer57} for the proof).
\begin{senumerate}
\item
\label{thm:baer}
$\Fit(G) = \set{g \in G}{\mcomm{h}{\omegaG}{g} =1\text{ for all } h\in G}.$
\end{senumerate}

Now we define the \emph{upper Fitting series}
\[1 = \cU_0(G) \Nle \cU_1(G) \Nle \cdots \Nle \cU_k (G) = G \]
by $\cU_{i+1}(G)/\cU_i(G) = \Fit(G/\cU_i(G))$.
If the group is clear, we simply write $\cU_i $ for $\cU_i(G)$.
The number of factors $k$ is called the \emph{Fitting length} of $G$ (denoted by $\FitL(G)$). 
The following fact can be derived by a straightforward induction from the characterization of $\Fit(G)$ as largest nilpotent normal subgroup.
\begin{senumerate}
\item
\label{lem:Fitting}
For $H \Nleq G$ and $g \in G$ we have
\begin{itemize}
		\item $\cU_i (H) = \cU_i \cap H$, for all $i$,
        \item $\FitL(H) \leq i$ if and only if $H \leq \cU_i$,
		\item $\FitL\Ncl{g} = i$ if and only if $g \in \cU_{i} \setminus \cU_{i-1}$.
\end{itemize}
\end{senumerate}

\paragraph{Equations in Groups.}
A \emph{term} (in the language of groups) is a word over an alphabet $\cX \cup \cX^{-1}$ where $\cX$ is a set of variables. A \emph{polynomial} over a group $\m G$ is a term where some of the variables are replaced by constants~-- \ie a word over the alphabet $G \cup \cX \cup \cX^{-1}$.
Since we are dealing with finite groups only, a symbol $X^{-1}\in \cX^{-1}$ for $X\in \cX$ can be considered as an abbreviation for $X^{\abs{G}-1}$.
We write $\polyA(x_1, \dots, x_n)$ or short $\polyA(\o x)$ for a polynomial (resp.\ term) $\polyA$ with variables from $\oneset{x_1, \dots, x_n}$. There is a natural composition of terms and polynomials: if $\polyB(x_1, \dots, x_n),\polyA_1, \dots, \polyA_n$ are polynomials (resp.\ terms), we write $\polyB(\polyA_1, \dots, \polyA_n)$ for the polynomials (resp.\ terms) obtained by substituting each occurrence of a variable $x_i$ by the polynomial (resp.\ term) $\polyA_i$.

A tuple $(g_1, \dots, g_n)\in G^n$ is a \emph{satisfying} assignment for $\polyA$
if $\polyA(g_1, \dots, g_n)=1$ in $G$.
The problems $\polsat{G}$ and $\poleqv{G}$ are as follows: for both of them the input is a polynomial $\polyA(x_1, \dots, x_n)$.
For $\polsat{G}$ the question is whether there \emph{exists} a satisfying assignment,
for $\poleqv{G}$ the question is whether \emph{all} assignments are satisfying.
Note here that these problems have many other names. For example in in \cite{Weiss20,GoldmannR02}, \polsat{} is denoted by  EQN-SAT and \poleqv{}  by EQN-ID.

\paragraph{Inducible subgroups.}
According to \cite{GoldmannR02}, we call a subset $S \sse G$ \emph{inducible} if
$S = \set{\polyA(g_1, \dots, g_n)}{g_1, \dots, g_n \in G}$
for some polynomial $\polyA(x_1, \dots, x_n)$ of $\m G$.
The importance of inducible subgroups lies in the observation that one can restrict variables in equations to inducible subgroups (simply by replacing each variable by the polynomial defining the inducible subgroup). This immediately gives the following lemma.

\begin{lemma}[{\,\!\cite[Lemma 8]{GoldmannR02}, \cite[Lemma 9,  10]{HorvathS11}}]\label{lem:inducibleEQN}
If $\m H$ is an inducible subgroup of $\m G$, then
\begin{itemize}
	\item $\polsat{H}$ is polynomial time many-one reducible to $\polsat{G}$,
	\item $\poleqv{H}$ is polynomial time many-one reducible to $\poleqv{G}$.
\end{itemize}	
\end{lemma}

We will use this lemma to restrict our consideration for an appropriate subgroup of the form
$\gamma_k(G)$. We will see that such subgroups are inducible.



\section{Proof of Theorem \ref{thm:mainIntro}}

The proof of the theorem is based on coding (by group polynomials) functions that imitate the behaviour of conjunctions.
Unfortunately, the lengths of such $n$-ary conjunction-like group polynomials are not bounded by any polynomial in $n$ and, therefore, they cannot be used to show \np-completeness of \polsat{}.\footnote{In fact, the mentioned \textit{AND}-weakness conjecture prevents the existence of such short -- polynomial size -- ``conjunction-like'' expressions. On the other hand, our construction also shows that the strongest version of the \textit{AND}-weakness conjecture~-- a $2^{\Omega(n)}$ lower bound~-- does not hold.}
However, the group polynomials we are going to produce have length bounded by
$2^{\Oh(n^{\frac{1}{d-1}})}$ where $d = \FitL(G)$.
Given such relatively short conjunction-like group polynomials we reduce graph coloring or \SAT,
depending on whether $\abs{G/H} \geq 3$ for a carefully chosen large subgroup $H$ of $\m G$.
In any case such reduction, together with the ETH, would give the lower bound
$2^{\Omega(\log^{d-1} \el)}$ for $\polsat{G}$.

\newcommand{\polyQ}[1]{\po{q}^{(#1)}}
\newcommand{\polyQstar}[1]{\po{\tilde q}^{(#1)}}

\newcommand{\FittingIndex}{\alpha}

To see how to produce such relatively short conjunction-like polynomials, we start
with the upper Fitting series of $\m G$
\[
1 = \cU_0 \Nle \cU_1 \Nle \cdots \Nle \cU_d  = G
\]
to go downwards along this series and consecutively carefully choose
$h_{\FittingIndex} \in \cU_{\FittingIndex} \setminus \cU_{\FittingIndex-1}$
on each level $\FittingIndex = d,d-1,\ldots,1$ of this sequence.
Then we get two different cosets
$\cU_{\FittingIndex-1}$ and $h_{\FittingIndex}\cdot{}\cU_{\FittingIndex-1}$
which are supposed to simulate false and true values, respectively.

The conjunction-like polynomials are based on the terms
$\polyQstar{k}(z,x_1, \dots, x_{k})$ and $\polyQ{k}(z,x_1, \dots, x_{k},w)$ for $k\geq 0$
defined by
\begin{align*}
\polyQstar0(z)
    &= z, \\
\polyQstar{k}(z, x_1, \dots, x_{k})
    &= \mcomm{\polyQstar{k-1}(z, x_1, \dots, x_{k-1})}{\omegaG}{x_{k}},
    && \text{for } k \geq 1, \text{ and }\\
\polyQ{k}(z, x_1, \dots, x_{k},w)
    &= \polyQstar{k+1}(z, x_1, \dots, x_{k},w),
    && \text{for } k \geq 0.
\end{align*}
Note that our definition of the $\polyQ{k}$'s immediately yields
\begin{align}
\polyQ{k+1}(z,x_1,\dots, x_{k},w,w)=\polyQ{k}(z,x_1,\dots, x_{k},w)
\label{qkww}
\end{align}

\noindent
The conjunction-like behaviour of the $\polyQ{k}$'s on the $\cU_{\FittingIndex}$-cosets
is precisely described in the following lemma.

\begin{lemma}
\label{lem:lowerPolys}
For any level $1 \leq \FittingIndex \leq d-1$
and $h_{\FittingIndex+1} \in \cU_{\FittingIndex+1} \setminus \cU_{\FittingIndex}$
there is some $h_{\FittingIndex} \in \cU_{\FittingIndex} \setminus \cU_{\FittingIndex-1}$
such that for each $k\in \N$ we have
\begin{align*}
\polyQ{k}(h_{\FittingIndex},x_1, \dots, x_k, h_{\FittingIndex+1}) \in
\begin{cases}
    h_{\FittingIndex}\cdot \cU_{\FittingIndex-1},
    &\text{if } x_i \in h_{\FittingIndex+1}\cdot \cU_{\FittingIndex} \text{ for all $i$,}
    \\
	\hphantom{h_{\FittingIndex}\cdot{}}\cU_{\FittingIndex-1},
    &\text{if } x_i \in \cU_{\FittingIndex} \text{ for some $i$}.
\end{cases}
\end{align*}
\end{lemma}

\begin{proof}	
\newcommand{\centralIndex}{\beta}
	
In this proof we may, without loss of generality,
factor out our group $\m G$ by $\cU_{\FittingIndex-1}$,
or equivalently assume that $\FittingIndex=1$.
This means that $\cU_{\FittingIndex} = \Fit(G)$ and so, by Baer's theorem (\ref{sec:prelims}.\ref{thm:baer}),
there is some $a \in G$ with $\mcomm{a}{\omegaG}{h_{\FittingIndex+1}} \neq 1$.
Let $\centralIndex \in \N$ be maximal such that
$\mcomm{a}{\omegaG}{h_{\FittingIndex+1}} \in \gamma_{\centralIndex}(\cU_{\FittingIndex})
\setminus  \{1\}$ for some $a \in G$.
Now, we simply put  $h_{\FittingIndex} = \mcomm{a}{\omegaG}{h_{\FittingIndex+1}}$,
to observe that $h_{\FittingIndex} = \mcomm{h_{\FittingIndex}}{\omegaG}{h_{\FittingIndex+1}}$
and $\Ncl{h_{\FittingIndex}}\leq \gamma_{\centralIndex}(\cU_{\FittingIndex})$.
The last inclusion gives that for all $x_1, \dots, x_{k+1}\in G$ we have
$\polyQ{k}(h_{\FittingIndex}, x_1, \dots, x_{k+1})\in \gamma_\centralIndex(\cU_{\FittingIndex})$.

Suppose now that one of the $x_i$'s is in $\cU_{\FittingIndex}$.
Then
$\polyQstar{i}(h_{\FittingIndex}, x_1, \dots, x_i)
= \mcomm{\polyQstar{i-1}(h_{\FittingIndex}, x_1, \dots, x_{i-1})}{\omegaG}{x_i}
\in \mcomm{\cU_{\FittingIndex}}{\omegaG}{\cU_{\FittingIndex}}
= \nilRes{\cU_{\FittingIndex}} = \{1\}$.
Hence, also $\polyQ{k}( h_{\FittingIndex}, x_1, \dots, x_k,h_{\FittingIndex+1})=1$.
	
On the other hand, if all the $x_i$'s are in the coset $h_{\FittingIndex+1} \cU_{\FittingIndex}$, then, by (\ref{sec:prelims}.\ref{lem:commutatorBasics}.\ref{iteratedCommIdright}), we have
$\polyQ{k}(h_{\FittingIndex},x_1, \dots, x_k, h_{\FittingIndex+1}) \equiv
\polyQ{k}(h_{\FittingIndex},h_{\FittingIndex+1}, \dots, h_{\FittingIndex+1}, h_{\FittingIndex+1})=
h_{\FittingIndex}$ modulo
$[\Ncl{h_{\FittingIndex}},\cU_{\FittingIndex}]
\leq [\gamma_{\centralIndex}(\cU_{\FittingIndex}) ,\cU_{\FittingIndex}]
\leq \gamma_{\centralIndex+1}(\cU_{\FittingIndex})$.
Hence, $\polyQ{k}(h_{\FittingIndex}, \o x,h_{\FittingIndex+1}) = h_{\FittingIndex} f$
for some $f \in \gamma_{\centralIndex+1}(\cU_{\FittingIndex})$.
Thus, all we have to show is that $f \in \cU_{\FittingIndex-1}$,
or -- in our setting -- that $f=1$.
To do this we induct on $j\geq \centralIndex+1$ to show that
$f \in \gamma_j(\cU_{\FittingIndex})$ for all $j$'s.

Starting with $f \in \gamma_j(\cU_{\FittingIndex}) \leq \gamma_{\centralIndex+1}(\cU_{\FittingIndex})$,
we also have
$\mcomm{f}{\omegaG}{h_{\FittingIndex+1}} \in \gamma_{\centralIndex+1}(\cU_{\FittingIndex})$.
But now, maximality of $\centralIndex$ ensures us that
\begin{align}
	\mcomm{f}{\omegaG}{h_{\FittingIndex+1}} =1.
\label{fdescent}
\end{align}
Obviously $[f,g] \in \gamma_{j+1}(\cU_{\FittingIndex})$ whenever
$f \in \gamma_j(\cU_{\FittingIndex})$ and $g \in \cU_{\FittingIndex}$.
This simply means that
$f\in C_{G/\gamma_{j+1}(\cU_{\FittingIndex})}
(\cU_{\FittingIndex}/\gamma_{j+1} (\cU_{\FittingIndex}))$,
and by (\ref{sec:prelims}.\ref{lem:commutatorBasics}.\ref{centralComm}) we obtain
\begin{align}
	\mcomm{h_{\FittingIndex}f}{\omegaG}{h_{\FittingIndex+1}}
	&\equiv \mcomm{h_{\FittingIndex}}{\omegaG}{h_{\FittingIndex+1}}\mcomm{f}{\omegaG}{h_{\FittingIndex+1}} \mod \gamma_{j+1}(\cU_{\FittingIndex}).\label{mcommdistributive}
\end{align}
Summing up we get
\begin{align*}
h_{\FittingIndex}f
&= \mcomm{h_{\FittingIndex}f}{\omegaG}{h_{\FittingIndex+1}}
\tag{by (\ref{qkww})}
\\[.2em]
&\equiv \mcomm{h_{\FittingIndex}}{\omegaG}{h_{\FittingIndex+1}}\mcomm{f}{\omegaG}{h_{\FittingIndex+1}} &&\hspace{-20mm}\mod \gamma_{j+1}(\cU_{\FittingIndex})
\tag{by (\ref{mcommdistributive})}
\\[.2em]
&= h_{\FittingIndex} \cdot 1,
\tag{by (\ref{fdescent})}
\end{align*}
so that $f\in \gamma_{j+1}(\cU_{\FittingIndex})$.

Going along the $j$'s we arrive to a conclusion that $f \in  \nilRes{\cU_{\FittingIndex}} = \{1\}$, as promised.
\myqed	
\end{proof}

Now, picking $h_d \in G\setminus\cU_{d-1}$, the consecutive use of  \autoref{lem:lowerPolys}
supplies us with elements $h_{d-1},\ldots,h_1$ that allow us to define conjunction-like polynomials
$\polyQ{k}_{\FittingIndex}(x_1, \dots, x_k)
=\polyQ{k}(h_{\FittingIndex-1},x_1, \dots, x_k, h_{\FittingIndex})$.
Note here that, since the terms $\polyQ{k}$ use iterated commutators ($\omegaG\cdot(k+2)$ times),
their sizes are exponential in $k$.
However, to get a conjunction on $n=k^{d-1}$ elements we first split these elements into
$k^{d-2}$ groups, each having $k$ elements.
If there were only two cosets of $\m G$ of $\cU_{d-1}$,
then applying to each such $k$ element group the polynomial $\polyQ{k}_d$ everything
would be send into $\cU_{d-2} \cup h_{d-1}\cdot\cU_{d-2}$. 
Now, we group the obtained $k^{d-2}$ values into $k^{d-3}$ groups, each of size $k$
and apply $\polyQ{k}_{d-1}$ to each such group.
Repeating this procedure we finally arrive into $\cU_1$ ensuring that the appropriate composition of the $\polyQ{k}_{\FittingIndex}$'s returns either the value $1$ or $h_1$.
One can easily notice that the size of such composed polynomial is
$2^{\Oh(k)} = 2^{\Oh(n^\frac{1}{d-1})}$.

Unfortunately, the behaviour of the $\polyQ{k}_d$'s and the entire long compositions can be controlled only on two cosets of $\cU_{d-1}$.
This requires $\abs{G/\cU_{d-1}} = 2$~-- which very seldom is the case.
Thus, the very top level requires a very careful treatment.
First, we replace the group $\m G$ with a smaller subgroup $\m G_0$ of the same Fitting length but such that $\m G_0$ is abelian over its $\cU_{d-1}$.
Then we find a normal subgroup  $\cU_{d-1}\leq H \Nle G_0$
so that we will be able to control the behaviour of the $\polyQ{k}$'s
on all cosets of $H$ in $\m G_0$.
The first step towards realizing this idea is described in the next observation.

\begin{lemma}
\label{lem:topQuotientabelian}
In each finite solvable group $\m G$ there is a subgroup $\m G_0$ satisfying:
\begin{itemize}
  \item $\m G_0$ is inducible,
  \item $\FitL(\m G_0)=\FitL(\m G)=d$, and
  \item $\m G_0/\cU_{d-1}(G_0)$ is abelian.
\end{itemize}
\end{lemma}

\begin{proof}
We simply set $\m G_0 = \gamma_m(G)$ where $m$ is maximal with $\gamma_m(G) \not\leq \cU_{d-1}(G)$.
This secures $\FitL(\m G_0)=d$.
To see that all groups $\gamma_j(G)$ in the lower central series
\[
G = \gamma_0(G) \geq \gamma_1(G) \geq \dots \geq \gamma_\omegaG(G)
\]
are inducible, we induct on $j$ and argue like in \cite[Lemma 5]{GoldmannR02}.
Let $\gamma_j(G)$ be the image of the polynomial $\po p(\o x)$.
Every element in $\gamma_{j+1}(G) = [\gamma_j(G), G]$ is a product of at most $\abs{G}$
elements of the form $[z,y]$, where $z$ ranges over $\gamma_j(G)$ and $y$ over entire $G$.
Thus, introducing new sequences of pairwise different variables $\o x^1,\ldots, \o x^{\abs{G}}$
we can produce $\gamma_{j+1}(G)$ as the image of the polynomial
$\prod_{i=1}^{\abs{G}} [\po p(\o x^i),y_i]$.

Finally, $\m G_0/\cU_{d-1}(G_0)$ is abelian as we have
$[G_0,G_0] = [\gamma_m(G),\gamma_m(G)] \leq [\gamma_m(G),G] =\gamma_{m+1}(G) \leq \cU_{d-1}$,
where the last inclusion is the consequence of the maximality of $m$.
\myqed	
\end{proof}

\newcommand{\gpK}{K}
\newcommand{\gpKm}{K_0}
\newcommand{\gpCen}{H}

From now on we simply change notation and replace our starting group $\m G$ by $\m G_0$,
or in other words we assume that $\m G/\cU_{d-1}(G)$ is abelian.
Now, to construct (and control) the promised normal subgroup $H$ first we pick  $\gpK \Nleq G$ among the minimal (with respect to inclusion) normal subgroups satisfying:
\begin{itemize}
	\item $[\gpK,G] = \gpK$ and
	\item $\FitL(\gpK) = d-1$.
\end{itemize}
Since $\nilRes{G}$ satisfies both above conditions, such $\gpK$ indeed exists.

\begin{senumerate}
\item
\label{lem:indecomposable}
$K$ is indecomposable, i.e. if $\gpK =\gpK_1\gpK_2$ for some $\gpK_1,\gpK_2 \Nleq G$ then $\gpK=\gpK_1$ or $\gpK=\gpK_2$.
\end{senumerate}

\begin{proof}
Suppose that $(\gpK_1,\gpK_2)$ is a minimal pair (coordinatewise) with $\gpK =\gpK_1\gpK_2$.
Since $\gpK = [\gpK,G] = [\gpK_1\gpK_2,G] = [\gpK_1,G] [\gpK_2,G]$
and $[\gpK_i,G] \leq \gpK_i$, we immediately get $[\gpK_i,G] = \gpK_i$ for both $i=1,2$.
Now if $\gpK_i< \gpK$, then minimality of $\gpK$ gives $\FitL(\gpK_i) \leq d-2$.
If this happens for both $\gpK_1$ and $\gpK_2$, then
$d-1 =\FitL(\gpK)= \FitL(\gpK_1\gpK_2) = \max\oneset{\FitL(\gpK_1),\FitL(\gpK_2)}\leq d-2$,
a contradiction.
\myqed
\end{proof}

By \sref{lem:indecomposable} we know that there exists the unique
$\gpKm \Nleq G$ with $\gpKm < \gpK$
and such that there is no normal subgroup of $\m G$ that lies strictly between $\gpKm$ and $\gpK$.

Note that, if $a \in \gpK\setminus \gpKm$, we cannot have $\Ncl{a}\leq \gpKm$. This gives
\begin{senumerate}
\item
\label{lem:principalNSGP}
For all $a \in \gpK\setminus \gpKm$ we have $\Ncl{a} = \gpK$.
\end{senumerate}

The other consequence of the fact that the solvable group $\m G$ has no normal subgroups
strictly between $\gpKm$ and $\gpK$ is the following.
\begin{senumerate}
\item
\label{lem:abelianQuotient}
$\gpK/\gpKm$ is abelian.
\end{senumerate}

We will also need:

\begin{senumerate}
\item
\label{lem:descent}
$\mcomm{\gpKm}{\omegaG}{G} \leq \cU_{d-2}(\gpK)$.
\end{senumerate}

\begin{proof}
By our choice of $\omegaG$, we have
$[\mcomm{\gpKm}{\omegaG}{G}, G ] = \mcomm{\gpKm}{\omegaG}{G}$.
Since $\mcomm{\gpKm}{\omegaG}{G}\leq \gpKm$ is strictly contained in $\gpK$
and $\gpK$ was chosen to be minimal with $[\gpK, G] = \gpK$ and $\FitL(\gpK) = d-1$,
we must have $\FitL(\mcomm{\gpKm}{\omegaG}{G}) \leq d-2$.
\myqed
\end{proof}

Now we are ready to define the normal subgroup $H$ of $\m G$.
We simply put $H$ to be the centralizer in $G$ of $\gpK$ modulo $\gpKm$,
i.e the largest normal subgroup with $[H,\gpK]\leq\gpKm$.
Then obviously $\gpCen = \set{g \in G}{[\gpK,g] \leq \gpKm}$.

\begin{senumerate}
\item
\label{lem:abelianCentralizerQuotient}
$\cU_{d-1} \leq \gpCen < G$.
In particular, $G/\gpCen$ is abelian.
\end{senumerate}

\begin{proof}
To see that $\gpCen < G$ suppose otherwise, i.e. $[\gpK,G] \leq \gpKm$.
This, however, contradicts our choice of $\gpK$ to satisfy $[\gpK,G] = \gpK$.

The first inclusion is simply equivalent to $[\gpK,\cU_{d-1}] \leq \gpKm$.
Indeed, since $\FitL(\gpK) = d-1$, we have $\mcomm{\gpK}{\omegaG}{\cU_{d-1}} \leq \nilRes{\cU_{d-1}} \leq \cU_{d-2}$ and, thus, $[\gpK,\cU_{d-1}]< \gpK$.
Since we assumed $G/\cU_{d-1}$ to be abelian, the second part of the statement follows.
\myqed
\end{proof}

Directly from our definitions, we know that $[x,y] \in \gpKm$ whenever $x\in \gpK$ and $y\in\gpCen$.
But the reason for our careful choice of $\gpK$ and then $\gpCen$ was to have a precise control
over the behaviour of $[x,y]$ for $y$ in other cosets of $\gpCen$ (and $x$ still in $\gpK$.)

Thus, for any $g \in G$ we define a map
$\phi_g: \gpK \to \gpK/\gpKm$ by $\phi_g(x) = [x,g]\cdot \gpKm$.
Since by \sref{lem:abelianQuotient}  is $\gpK/\gpKm$ is abelian,
using (\ref{sec:prelims}.\ref{lem:commutatorBasics}.\ref{basicCommId}), one can easily check that
$\phi_g$ is a group homomorphism for all $g \in G$.
Also we have $\phi_g(\gpKm) \leq \gpKm$, i.e. the kernel of this homomorphism contains $\gpKm$
so that $\phi_g$ actually induces a homomorphism $\gpK/\gpKm \to \gpK/\gpKm$.
We also write $\phi_g$ for this induced homomorphism.

\begin{senumerate}
\item
\label{lem:phiGiso}
If $g \in G \setminus \gpCen$, then $\phi_g: \gpK/\gpKm \to \gpK/\gpKm$ is an isomorphism. 
\end{senumerate}

\begin{proof}
We start with showing that for $g\in G$
\begin{align}
	\phi_g(x^b) = \phi_g(x)^b \label{eq:moduleHom}
\end{align}
whenever $x\in\gpK$ and $b\in G$.
Indeed, by \sref{lem:abelianCentralizerQuotient},
we can write $bg = h gb $ for some $h \in \gpCen$.
Then we have
\begin{align*}
	\phi_g(x^b) 	&= [x^b,g]\cdot \gpKm\\
	& = (x^b)^{-1} \, g^{-1}  b^{-1} x bg\, \cdot\gpKm\\
	& = (x^b)^{-1}\,  b^{-1} g^{-1}h^{-1} x hgb \cdot\gpKm\\
	& = (x^b)^{-1} \, b^{-1} g^{-1}x gb\cdot \gpKm \tag{since $h \in \gpCen$}\\
	& = (x^{-1} \,  g^{-1}x g)^b\cdot\gpKm \\
	&= \phi_g(x)^b.
\end{align*}

To see that the kernel of the original $\phi_g$ is $\gpKm$,
pick $a \in \gpK\setminus \gpKm$, so that, by \sref{lem:principalNSGP},
every element $x \in \gpK$ can be represented as $x= a^{g_1} \cdots a^{g_n}$
for some $g_1, \dots, g_n \in G$.
Now, if $\phi_g(a) =\gpKm$, then \eqref{eq:moduleHom} gives
$\phi_g(x) =\gpKm$ for all $x \in \gpK$.
This would however put $g$ into the centralizer $\gpCen$, contrary to our assumption.
\myqed
\end{proof}
	
Note that \eqref{eq:moduleHom} means that $\phi_g$ is not only a group homomorphism
but actually a homomorphism of $G$-modules.
Here $\gpK/\gpKm$ is a $G$-module under the action of $G$ on $\gpK/\gpKm$ via conjugation.
In terms of modules the proof of \sref{lem:phiGiso} is stated even easier:
The kernel of $\phi_g$ has to be a submodule of $ \gpK/\gpKm$.
However, by \sref{lem:principalNSGP}
$\gpK/\gpKm$ is generated, as a $G$-module, by any of its non-trivial elements.

\begin{remark} Notice, that for \sref{lem:phiGiso}, we need $G/H$ to be abelian.
Indeed, in general, if $N$ is a minimal (and, thus, indecomposable) normal subgroup with $[N,G] = N$, the map $N \to N$ defined by $x \mapsto[x,g]$ is \emph{not} necessarily bijective for all $g \not\in C_G(N)$. For instance take the semidirect product $(C_3 \times C_3) \rtimes D_4$ where $D_4 = \genr{a,b}{a^2=b^2=(ab)^4=1}$ is the dihedral group of order 8 and $a$ acts by exchanging the two components of $C_3 \times C_3$ and $b$ by inverting the second one. Then, $N = C_3 \times C_3$ is an indecomposable normal subgroup and $[N,G] = N$ but $a \not\in C_G(N)$ and $[(1,1),a] = [(2,2),a] = 1$, so  $x \mapsto[x,g]$ is not bijective on $N$ (here we use an additive notation for $C_3 = \oneset{0,1,2}$).
\end{remark}

We summarize our observations in the following claim.
\begin{senumerate}
\item
\label{lem:phiomega}
For all $x \in \gpK$ we have
\begin{align*}
\polyQstar1(x,y) \in
\begin{cases}
	x \gpKm,  &\text{if } y \not\in\gpCen,\\
	  \gpKm,  &\text{if } y \in\gpCen.
\end{cases}
\end{align*}
\end{senumerate}

\begin{proof}
Note first that $\omegaG$ was chosen to satisfy $\mcomm{x}{\omegaG}{y} = \mcomm{x}{2\omegaG}{y}$.
Moreover, for a fixed $g \in G$ the unary polynomial $\polyQstar1(x,g)$ acts on $\gpK$ as the
composition $\phi_g^\omegaG$ of $\phi_g$ with itself $\omegaG$ times.
Now, if $g \not\in \gpCen$, then \sref{lem:phiGiso} yields that $\phi_g^\omegaG$ is the identity
on the quotient $\gpK/\gpKm$.
Moreover, $\phi_g^\omegaG$ is constant $\gpKm$ for $g\in \gpCen$.
\myqed
\end{proof}

With claim \sref{lem:phiomega} we are ready to construct polynomials that will allow to code
coloring or \SAT at the very top level.

\begin{lemma}
\label{prop:nicePolynomials}
There is $h \in \gpK\setminus \cU_{d-2}$ and families of polynomials
	\begin{align*}
		&\polyAND^{(k)}(y_1, \dots, y_k) \qquad\qquad\qquad\qquad\qquad\qquad \text{ and }\\
		&\polySAT^{(k)}(y_{1,1},y_{1,2},y_{1,3} \dots, y_{k,1},y_{k,2},y_{k,3})
	\end{align*}	
	of length $2^{\Oh(k)}$ such that	
	\begin{align}
		\polyAND^{(k)}(\vec y)  &\in
        \begin{cases}
			h \cdot \cU_{d-2}, & \text{if } y_i \not\in \gpCen \text{ for all $i$},\\
			\hphantom{h\cdot{}}\cU_{d-2}, & \text{if } y_i \in \gpCen \text{ for some $i$},
		\end{cases}\label{eq:ColorPolyUpper}
\intertext{and}
		\polySAT^{(k)}(\vec y)  &\in
        \begin{cases}
			h \cdot \cU_{d-2}, & \text{if for all $i$ there is some $j$ with } y_{i,j} \in \gpCen,\\
			\hphantom{h\cdot{}}\cU_{d-2},
                & \text{if } y_{i,1},y_{i,2},y_{i,3} \not\in \gpCen \text{ for some $i$}.
        \end{cases}\label{eq:3SATPolyUpper}
\end{align}
\end{lemma}

\begin{proof}
First, we use \sref{lem:phiomega} and induct on $k$
in order to see that for all $a \in \gpK \setminus \gpKm$ we have
\begin{align*}
	\polyQstar{k}(a,y_1, \dots, y_k) \in
	\begin{cases}
		a \gpKm,  &\text{if } y_i \not\in\gpCen \text{ for all $i$},\\
		  \gpKm,  &\text{if } y_i \in\gpCen \text{ for some $i$}.
	\end{cases}
\end{align*}
Now we fix some arbitrary $a \in \gpK \setminus \gpKm$ and $g \in G \setminus \gpCen$.
Then obviously also $h = \mcomm{a}{\omegaG}{g}\in \gpK \setminus \gpKm$.
Actually $h\not\in \cU_{d-2}$, as otherwise
$h\in \cU_{d-2} \cap \gpK \leq \gpKm$.

Now, by \sref{lem:descent} we know that $M \coloneqq \mcomm{\gpKm}{\omegaG}{G}\leq \cU_{d-2}(\gpK)$. By (\ref{sec:prelims}.\ref{lem:commutatorBasics}.\ref{iteratedCommId}) it follows that
\begin{align*}
	\polyQ{k}(a,y_1, \dots, y_k,g) \in \begin{cases}
		h M, &\text{if } y_i \not\in\gpCen \text{ for all $i$},\\
		M,  &\text{if } y_i \in\gpCen \text{ for some $i$}.
	\end{cases}
\end{align*}
Thus, $\polyAND^{(k)}(y_1, \dots, y_k) = \polyQ{k}(a,y_1, \dots, y_k,g) $ satisfies (\ref{eq:ColorPolyUpper}). Clearly, its length is in $2^{\Oh(k)}$.

\bigskip
To construct the polynomials $\polySAT^{(k)}$, we first define
\[
\polyD(x,y_1, y_2,y_3) = x\cdot \polyQstar3(x, y_1,y_2,y_3)^{-1}.
\]
Then for all $x \in K$, by \sref{lem:phiomega}, we have
\begin{align*}
	\polyD(x,y_1, y_2,y_3) \in
    \begin{cases}
	      \gpKm, &\text{if } y_j \not\in\gpCen \text{ for all $j$},\\
		x \gpKm, &\text{if } y_j \in\gpCen \text{ for some $j$}.
	\end{cases}
\end{align*}
Now, with $a,g,h$ and $M$ as above, we proceed as with the $\polyAND^{(k)}$'s to define
%
%
\begin{align*}
\tilde\polySAT^{(k)}(\o y)
&= \polyD(\cdots \polyD(a, y_{1,1},y_{1,2},y_{1,3}),\dots ,y_{k,1},y_{k,2},y_{k,3})
\intertext{and}
\polySAT^{(k)}(\o y) &= \mcomm{\vphantom{k^k}\tilde\polySAT^{(k)}(\o y)}{\omegaG}{g}.
\end{align*}
As previously, \eqref{eq:3SATPolyUpper} follows from (\ref{sec:prelims}.\ref{lem:commutatorBasics}.\ref{iteratedCommId}).
\myqed
\end{proof}

Our next claim summarizes \autoref{lem:lowerPolys} and \autoref{prop:nicePolynomials}.

\begin{lemma}
\label{lem:finalPolynomials}
For $1 \leq \FittingIndex \leq d-1$ there are elements $h_{\FittingIndex}\neq 1$
and families of polynomials
\begin{align*}
	 	&\polyAND_{\FittingIndex}^{(m)}(y_1, \dots, y_m)
                \qquad\qquad\qquad\qquad\qquad\qquad \text{ and }\\
	 	&\polySAT_{\FittingIndex}^{(m)}(y_{1,1},y_{1,2},y_{1,3} \dots, y_{m,1},y_{m,2},y_{m,3})
\end{align*}
of length $2^{\Oh(m^{\frac{1}{d-\FittingIndex}})}$ such that	
\begin{align*}
\polyAND_{\FittingIndex}^{(m)}(\vec y)  &\in
\begin{cases}
	h_{\FittingIndex}\cdot \cU_{\FittingIndex-1},
        & \text{if } y_i \not\in \gpCen \text{ for all $i$},\\
    \hphantom{h_{\FittingIndex}\cdot{}}\cU_{\FittingIndex-1},
        & \text{if } y_i \in \gpCen \text{ for some $i$},
\end{cases}
\intertext{and}
\polySAT_{\FittingIndex}^{(m)}(\vec y)  &\in
\begin{cases}
	h_{\FittingIndex}\cdot \cU_{\FittingIndex-1},
        & \text{if for all $i$ there is some $j$ with } y_{i,j} \in \gpCen,\\
	\hphantom{h_{\FittingIndex}\cdot{}}\cU_{\FittingIndex-1},
        & \text{if } y_{i,1},y_{i,2},y_{i,3} \not\in \gpCen \text{ for some $i$}.
\end{cases}
\end{align*}
\end{lemma}

\begin{proof}
We induct downwards on $\FittingIndex=d-1,\ldots,2,1$.
To start with we refer to \autoref{prop:nicePolynomials}
to set $h_{d-1} = h$ while
$\polyAND_{d-1}^{(m)}(\vec y)=\polyAND^{(m)}(\vec y)$
and $\polySAT_{d-1}^{(m)}(\vec y)=\polySAT^{(m)}(\vec y)$.

Now let $\FittingIndex<d-1$ and set $k = \ceil{\sqrt[d-\FittingIndex]{m}}$
and $\ell = \ceil{\frac{m}{k}}$.
By possibly duplicating some of the variables we may assume that $m = k\ell$.

To define $\polyAND_{\FittingIndex}^{(m)}(\o y)=\polyAND_{\FittingIndex}^{(m)}(y_1, \dots, y_m)$
we first refer to \autoref{lem:lowerPolys} to get $h_{\FittingIndex}$ from $h_{\FittingIndex+1}$
and then we set
\begin{align*}
\polyAND_{\FittingIndex}^{(m)}(\o y) &= \polyQ{k}\!\left(h_{\FittingIndex},\polyAND_{\FittingIndex+1}^{(\ell)}(y_1,\dots ,y_\ell), \dots, \polyAND_{\FittingIndex+1}^{(\ell)}(y_{m-\ell+1},\dots ,y_m), h_{\FittingIndex+1}\right),
\end{align*}
where the polynomial $\polyAND_{\FittingIndex+1}^{(\ell)}$ is supplied by the induction hypothesis.
From \autoref{lem:lowerPolys} it should be clear that $\polyAND_{\FittingIndex}^{(m)}$
satisfies the condition claimed for it.

Also its length can be bounded inductively.
Substituting to the polynomial
$\polyQ{k}(h_{\FittingIndex}, x_1, \dots, x_k,h_{\FittingIndex+1})$
of length $2^{\Oh(k)}$ (by \autoref{lem:lowerPolys})
the $k= m^{\frac{1}{d-\FittingIndex}}$ copies of the polynomial
$\polyAND_{\FittingIndex+1}^{(\ell)}$ of length
$2^{\Oh\bigl(\ell^{\frac{1}{d-\FittingIndex-1}}\bigr)}$
and using $\ell = m^{\frac{d-\FittingIndex-1}{d-\FittingIndex}}$
we arrive at the following bound for the length of $\polyAND_{\FittingIndex}^{(m)}$
\begin{align*}
2^{\Oh(k)}\cdot 2^{\Oh(\ell^{\frac{1}{d-\FittingIndex-1}})} &= 2^{\Oh\left(m^{\frac{1}{d-\FittingIndex}} +   \left(m^{\frac{d-\FittingIndex-1}{d-\FittingIndex}}\right)^{\frac{1}{d-\FittingIndex-1}}\right)}\\
	&=2^{\Oh\left(m^{\frac{1}{d-\FittingIndex}} +   m^{\frac{d-\FittingIndex-1}{d-\FittingIndex}\cdot \frac{1}{d-\FittingIndex-1}}\right)}\\
	&=2^{\Oh\left(m^{\frac{1}{d-\FittingIndex}}\right)}.
\end{align*}

In a very similar  way we produce $\polySAT_{\FittingIndex}^{(m)}(\vec y)$
from the $\polySAT_{\FittingIndex+1}^{(\ell)}$'s
by simply putting
\begin{align*}
\polySAT_{\FittingIndex}^{(m)}(\vec y) &= \polyQ{k}\big(h_{\FittingIndex},\polySAT_{\FittingIndex+1}^{(\ell)}(y_{1,1},y_{1,2},y_{1,3}, \dots, y_{\ell,1},y_{\ell,2},y_{\ell,3}),\ldots\\ & \hspace{-1em}\dots, \polySAT_{\FittingIndex+1}^{(\ell)}(y_{m-\ell+1,1},y_{m-\ell+1,2},y_{m-\ell+1,3}, \dots, y_{m,1},y_{m,2},y_{m,3}),h_{\FittingIndex+1}\big).\tag*{\myqed}	
\end{align*}
\end{proof}

\newcommand{\polyGraph}{\polyAND_{\mbox{{\kern-.12em\fontsize{6}{6}\selectfont $\Gamma$}}}}

Now we are ready do conclude our proof of Theorem \ref{thm:mainIntro}.
Recall that due to \autoref{lem:topQuotientabelian} we are working in the group $\m G$
in which $G/\cU_{d-1}G$ is abelian.
We are going to reduce \SAT or \KColoring{C} to $\polsat{G}$ and $\poleqv{G}$
depending on whether $C=\abs{G/\gpCen}>2$ or not.
In either case the reduction from \KColoring{C} to $\polsat{G}$ and $\poleqv{G}$ works; however, the case $C=2$ has to be treated
in a different way since $\KColoring{2}$ is decidable in polynomial time.

In our reduction the formula $\Phi$ from \SAT (or a graph $\Gamma$ from \KColoring{C})
is transformed to a polynomial $\polySAT_\Phi$ (or $\polyGraph$)
and a group element $h_1$ so that the following will hold:
\begin{enumerate}[(A)]
\item \label{pointA}
the length of $\polySAT_\Phi$ (resp.\ $\polyGraph$) is in $2^{\Oh(\sqrt[d-1]{m})}$
where $m$ is the number of clauses (resp.\ the number of edges),
\item \label{pointB}
$\polySAT_\Phi$ (resp.\ $\polyGraph$) can be computed in time $2^{\Oh(\sqrt[d-1]{m})}$ (\ie polynomial in the length of $\polySAT_\Phi$ (resp.\ $\polyGraph$)), 
\item \label{pointC} if $\Phi$ is satisfiable (resp.\ $\Gamma$ has a valid $C$-coloring), then
$\polySAT_\Phi = h_1$  (resp.\ $\polyGraph =  h_1$)  is satisfiable, and,
\item \label{pointD}
if $\Phi$ is \emph{not} satisfiable (resp.\ $\Gamma$ does \emph{not} have a valid $C$-coloring), then
$\polySAT_\Phi = 1$  (resp.\ $\polyGraph = 1$) holds under \emph{all} evaluations.
\end{enumerate}
The latter two points imply that $\polySAT_\Phi = h_1$  (resp.\ $\polyGraph =  h_1$)  is satisfiable
if and only if $\Phi$ is satisfiable (resp.\ $\Gamma$ has a valid $C$-coloring) and
$\polySAT_\Phi = 1$  (resp.\ $\polyGraph = 1$) holds identically in  $\m G$ 
if and only if $\Phi$ is \emph{not} satisfiable (resp.\ $\Gamma$ does \emph{not} have a valid $C$-coloring).

Now, if $\ell$ denotes the input length for \polsat{} or \poleqv{}
(i.e. the size of $\polySAT_\Phi$ or $\polyGraph$),
then an algorithm for \polsat{} or \poleqv{} working in $2^{o(\log^{d-1}\ell)}$-time
would solve \SAT (resp.\ \KColoring{C}) in time
\[
2^{\Oh(\sqrt[d-1]{m})} + 2^{o(\log^{d-1}(2^{\sqrt[d-1]{m}}))} = 2^{o(m)},
\]
contradicting ETH.

We start with describing the reduction from \KColoring{C} to $\polsat{G}$ and $\poleqv{G}$ where $C=\abs{G/\gpCen}$. The quotient $\abs{G/\gpCen}$
serves as the set of colors.
For a graph $\Gamma = (V,E) $ with $E \sse\binom{V}{2}$, $\abs{V} = n$ and $\abs{E} = m$,
we use variables $x_v$ for $v \in V$. For an edge $\{u,v\}\in E$ the value of $x_{u}x_{v}^{-1}$ (modulo $\gpCen$) decides whether the vertices $u,v$ have the same color.
To control whether the coloring of $\Gamma$ is proper we define the polynomial $\polyGraph$ by putting
\[
\polyGraph((x_v)_{v\in V}) = \polyAND_1^{(m)}\!\!\left((x_{u}x_{v}^{-1})_{\{u,v\}\in E}\right)
\]
where $\polyAND_1^{(m)}$ and $h_1$ are supplied by \autoref{lem:finalPolynomials}~-- and, thus, meet the length bound (\ref{pointA}). Point (\ref{pointB}) is clear from the definition of the polynomial.
Notice that the edges can be fed into $\polyAND_1^{(m)}$
in any order without affecting the final value of such polynomials.
Every evaluation of the variables $x_v$ by elements of $\m G$
defines a coloring $\chi : V \to G/\gpCen$ in a natural way.
If this coloring is valid (i.e. $\chi(u) \not\equiv \chi(v)\mod \gpCen$ for every edge $\{u,v\}\in E$),
then all the expressions $\chi(u) \chi(v)^{-1}$ are not in $\gpCen$ and \autoref{lem:finalPolynomials}
ensures us that $\polyGraph((x_v)_{v\in V}) = h_1$. This shows (\ref{pointC}). 

Conversely, by \autoref{lem:finalPolynomials}, for every evaluation of the $x_v$'s by elements of $\m G$
that does not satisfy the equation $\polyGraph((x_v)_{v\in V}) = 1$, we have $x_u x_v^{-1} \not\in \gpCen$ for all edges $\{u,v\}$. This obviously yields a valid coloring of $\Gamma$~-- hence, it proves (\ref{pointD}).

\newcommand{\polyClause}{\po{c}}
\medskip
As \KColoring{2} is solvable in polynomial time in the case $\abs{G/\gpCen} =2$,
we interpret \SAT and use the two cosets of $\gpCen$ in $\m G$ as the true/false boolean values.
We start with the formula
\[
\Phi = (A_{1,1} \lor A_{1,2}\lor A_{1,3}) \land \cdots \land (A_{m,1} \lor A_{m,2}\lor A_{m,3}),
\]
where each literal $A_{i,j}$ is either one of the boolean variables $X_1, \dots, X_n$ or its negation.
First, we transform the literals $A_{i,j}$ into the expressions $x_{i,j}$ that are supposed to range over $\m G$ by picking $g\in G\setminus\gpCen$ and then setting
\[
x_{i,j} =
\begin{cases}
	gx_k, &\text{if } A_{i,j} = X_k,\\
	x_k, &\text{if } A_{i,j} = \lnot X_k.
\end{cases}
\]
Finally, we set
\begin{align*}
\polySAT_\Phi(x_1,\ldots,x_n)
&= \polySAT_1^{(m)}\!\left(x_{1,1},x_{1,2},x_{1,3}, \dots,x_{m,1},x_{m,2},x_{m,3}\right)
\end{align*}
where again $\polySAT_1^{(m)}$ is supplied by  \autoref{lem:finalPolynomials}.

Now, given an assignment to the boolean variables $X_1, \dots, X_n$, we obtain an assignment for $x_1, \dots, x_n$ by setting $x_i= g$ if $X_i$ is \emph{true} and $x_i= 1$ if $X_i$ is \emph{false}. It can be easily checked using \autoref{lem:finalPolynomials} that the original assignment was satisfying for $\Phi$ if and only if $\polySAT_\Phi(\vec x) = h_1$ is satisfied (notice that $g^2 \in H$). This shows (\ref{pointC}).
On the other hand, if $\polySAT_\Phi(\vec x) \neq 1$, then, by \autoref{lem:finalPolynomials}, for all $i$ there is some $j$ with  $x_{i,j} \in H$. Hence, if we assign \emph{true} to $X_k$ if and only if $x_k \not\in H$, we obtain a satisfying assignment for $\Phi$~-- proving (\ref{pointD}).
\myqed

\section{Conclusion}

With \autoref{thm:mainIntro} in mind, one could suspect that finite solvable groups of Fitting length 2 have polynomial time algorithms for $\polsat{}$.
As we have already mentioned, the very recent paper \cite{FoldvariH19} shows that $\polsat{}$ is in \ptime for many such groups, in particular, for all semidirect products $\m G_p \rtimes \m A$, where $\m G_p$ is a $p$-group and $\m A$ is abelian.
This, however, does not cover e.g. the dihedral group $\m D_{15}$.
In fact, in \cite{IdziakKKD15} $\polsat{D_{15}}$ is shown not to be in \ptime, unless ETH fails.
On the other hand, $\poleqv{D_{15}}\in \ptime$.
Actually from \cite{FoldvariH19} we know that $\poleqv{G} \in \ptime$ for each semidirect product
$\m G = \m N \rtimes \m A$ where $\m N$ is nilpotent and $\m A$ is abelian.
In fact, $\m D_{15}$ is the first known example of a group with polynomial time \poleqv{} and non-polynomial (under ETH) \polsat{}.
The converse situation cannot happen
as, for a group $\m G$, $\polsat{G}\in \ptime$ implies $\poleqv{G}\in \ptime$.
Indeed, to confirm that $\po t(\o x)=1$ holds for all possible values of the $\vec x$'s, we check that for no $g\in G \setminus \{1\}$ the equation $\po t(\vec x)=g$ has a solution.

We conclude our paper with two obvious questions.

\begin{problem}
	Characterize finite solvable groups (of Fitting length 2) with \polsat{}  decidable in polynomial time.
\end{problem}

\begin{problem}
	Characterize finite solvable groups (of Fitting length 2) with \poleqv{}  decidable in polynomial time.
\end{problem}

Finally, we want to point out the consequences of our main result to another problem: For a finitely generated (but possibly infinite) group with a finite set of generators $\Sigma$ the power word problem is as follows: The input is a tuple $(p_1, x_1, p_2, x_2, \ldots, p_n, x_n)$ where the $p_i$ are words over $\Sigma$ and the $x_i$ are integers encoded in binary. The question is whether $p_1^{x_1} \cdots p_n^{x_n}$ evaluates to the identity of the group. The complexity of the power word problem in a wreath product $\m G \wr \Z$ where $\m G$ is a finite group has a similar behaviour as \poleqv{}: if $\m G$ is nilpotent, the power word problem of $\m G \wr \Z$ is in polynomial time \cite{FigeliusGLZ20} (actually even in \TC) and, if $\m G$ is non-solvable, it is \coNP-complete \cite{LohreyW19}. Indeed,
in \cite{FigeliusGLZ20} a surprising connection to \poleqv{} has been pointed out: if $\m G$ is a finite group, then \poleqv{G} can be reduced in polynomial time to the power word problem of the wreath product $\m G \wr \Z$. In particular, \autoref{thm:mainIntro} implies that the power word problem of $G \wr \Z$ where $G$ is a finite solvable group of Fitting length at least three is not in \P assuming ETH.

\bibliographystyle{spmpsci}
\bibliography{equations}

\begin{thebibliography}{10}
\providecommand{\url}[1]{{#1}}
\providecommand{\urlprefix}{URL }
\expandafter\ifx\csname urlstyle\endcsname\relax
  \providecommand{\doi}[1]{DOI~\discretionary{}{}{}#1}\else
  \providecommand{\doi}{DOI~\discretionary{}{}{}\begingroup
  \urlstyle{rm}\Url}\fi

\bibitem{Baer57}
Baer, R.: Engelsche elemente {N}oetherscher {G}ruppen.
\newblock Math. Ann. \textbf{133}, 256--270 (1957).
\newblock \doi{10.1007/BF02547953}.
\newblock \urlprefix\url{https://doi.org/10.1007/BF02547953}

\bibitem{BarringtonMMTT00}
Barrington, D.A.M., McKenzie, P., Moore, C., Tesson, P., Th{\'e}rien, D.:
  Equation satisfiability and program satisfiability for finite monoids.
\newblock In: Mathematical Foundations of Computer Science 2000, 25th
  International Symposium, {MFCS} 2000, Proceedings, \emph{Lecture Notes in
  Computer Science}, vol. 1893, pp. 172--181. Springer (2000).
\newblock \doi{10.1007/3-540-44612-5\_13}.
\newblock \urlprefix\url{https://doi.org/10.1007/3-540-44612-5\_13}

\bibitem{BurrisL04}
Burris, S., Lawrence, J.: Results on the equivalence problem for finite groups.
\newblock Algebra Universalis \textbf{52}(4), 495--500 (2005).
\newblock \doi{10.1007/s00012-004-1895-8}.
\newblock \urlprefix\url{https://doi.org/10.1007/s00012-004-1895-8}

\bibitem{CyganFKLMPPS15}
Cygan, M., Fomin, F.V., Kowalik, L., Lokshtanov, D., Marx, D., Pilipczuk, M.,
  Pilipczuk, M., Saurabh, S.: Parameterized Algorithms.
\newblock Springer (2015).
\newblock \doi{10.1007/978-3-319-21275-3}.
\newblock \urlprefix\url{https://doi.org/10.1007/978-3-319-21275-3}

\bibitem{DiekertE17icalpshort}
Diekert, V., Elder, M.: Solutions of twisted word equations, {EDT0L} languages,
  and context-free groups.
\newblock In: ICALP 2017, Proceedings, \emph{LIPIcs}, vol.~80, pp. 96:1--96:14.
  Dagstuhl, Germany (2017).
\newblock \doi{10.4230/LIPIcs.ICALP.2017.96}.
\newblock \urlprefix\url{http://drops.dagstuhl.de/opus/volltexte/2017/7397}

\bibitem{FigeliusGLZ20}
Figelius, M., Ganardi, M., Lohrey, M., Zetzsche, G.: The complexity of knapsack
  problems in wreath products.
\newblock In: 47th International Colloquium on Automata, Languages, and
  Programming, {ICALP} 2020, July 8-11, 2020, Saarbr{\"{u}}cken, Germany
  (Virtual Conference), pp. 126:1--126:18 (2020).
\newblock \doi{10.4230/LIPIcs.ICALP.2020.126}.
\newblock \urlprefix\url{https://doi.org/10.4230/LIPIcs.ICALP.2020.126}

\bibitem{Foldvari18}
F{\"o}ldv{\'a}ri, A.: The complexity of the equation solvability problem over
  nilpotent groups.
\newblock J. Algebra \textbf{495}, 289--303 (2018).
\newblock \doi{10.1016/j.jalgebra.2017.10.002}.
\newblock \urlprefix\url{https://doi.org/10.1016/j.jalgebra.2017.10.002}

\bibitem{FoldvariH19}
F{\"o}ldv{\'a}ri, A., Horv{\'a}th, G.: The complexity of the equation
  solvability and equivalence problems over finite groups.
\newblock International Journal of Algebra and Computation \textbf{30}(03),
  607--623 (2020).
\newblock \doi{10.1142/S0218196720500137}.
\newblock \urlprefix\url{https://doi.org/10.1142/S0218196720500137}

\bibitem{GarretaMO20}
Garreta, A., Miasnikov, A., Ovchinnikov, D.: Diophantine problems in solvable
  groups.
\newblock Bulletin of Mathematical Sciences  (2020).
\newblock \doi{10.1142/S1664360720500058}

\bibitem{GoldmannR02}
Goldmann, M., Russell, A.: The complexity of solving equations over finite
  groups.
\newblock Inf. Comput. \textbf{178}(1), 253--262 (2002).
\newblock \doi{10.1006/inco.2002.3173}.
\newblock \urlprefix\url{https://doi.org/10.1006/inco.2002.3173}

\bibitem{Horvath11}
Horv{\'a}th, G.: The complexity of the equivalence and equation solvability
  problems over nilpotent rings and groups.
\newblock Algebra Universalis \textbf{66}(4), 391--403 (2011).
\newblock \doi{10.1007/s00012-011-0163-y}.
\newblock \urlprefix\url{https://doi.org/10.1007/s00012-011-0163-y}

\bibitem{Horvath15}
Horv{\'a}th, G.: The complexity of the equivalence and equation solvability
  problems over meta-{A}belian groups.
\newblock J. Algebra \textbf{433}, 208--230 (2015).
\newblock \doi{10.1016/j.jalgebra.2015.03.015}.
\newblock \urlprefix\url{https://doi.org/10.1016/j.jalgebra.2015.03.015}

\bibitem{HorvathM07}
Horv{\'a}th, G., M{\'e}rai, L., Szab{\'o}, C., Lawrence, J.: The complexity of
  the equivalence problem for nonsolvable groups.
\newblock Bulletin of the London Mathematical Society \textbf{39}(3), 433--438
  (2007).
\newblock \doi{10.1112/blms/bdm030}

\bibitem{HorvathS11}
Horv{\'a}th, G., Szab{\'o}, C.: The extended equivalence and equation
  solvability problems for groups.
\newblock Discrete Math. Theor. Comput. Sci. \textbf{13}(4), 23--32 (2011)

\bibitem{HorvathS06}
Horv{\'a}th, G., Szab{\'o}, C.A.: The complexity of checking identities over
  finite groups.
\newblock {IJAC} \textbf{16}(5), 931--940 (2006).
\newblock \doi{10.1142/S0218196706003256}.
\newblock \urlprefix\url{https://doi.org/10.1142/S0218196706003256}

\bibitem{Huppert67}
Huppert, B.: Endliche {G}ruppen. {I}.
\newblock Die Grundlehren der Mathematischen Wissenschaften, Band 134.
  Springer-Verlag, Berlin-New York (1967)

\bibitem{IdziakKK20}
Idziak, P.M., Kawa\l{}ek, P., Krzaczkowski, J.: Intermediate problems in
  modular circuits satisfiability.
\newblock In: Proceedings of the 35th Annual ACM/IEEE Symposium on Logic in
  Computer Science, LICS '20, p. 578–590. Association for Computing
  Machinery, New York, NY, USA (2020).
\newblock \doi{10.1145/3373718.3394780}.
\newblock \urlprefix\url{https://doi.org/10.1145/3373718.3394780}

\bibitem{IdziakKKD15}
Idziak, P.M., Kawa\l{}ek, P., Krzaczkowski, J.: Solving equations over dihedral
  groups (2020).
\newblock Manusript

\bibitem{IdziakK18}
Idziak, P.M., Krzaczkowski, J.: Satisfiability in multi-valued circuits.
\newblock In: Proceedings of the 33rd Annual ACM/IEEE Symposium on Logic in
  Computer Science, LICS '18, p. 550–558. Association for Computing
  Machinery, New York, NY, USA (2018).
\newblock \doi{10.1145/3209108.3209173}.
\newblock \urlprefix\url{https://doi.org/10.1145/3209108.3209173}

\bibitem{ImpagliazzoPZ01}
Impagliazzo, R., Paturi, R., Zane, F.: Which problems have strongly exponential
  complexity?
\newblock J. Comput. Syst. Sci. \textbf{63}(4), 512--530 (2001).
\newblock \doi{10.1006/jcss.2001.1774}.
\newblock \urlprefix\url{https://doi.org/10.1006/jcss.2001.1774}

\bibitem{KawalekK}
Kawa{\l}ek, P., Krzaczkowski, J.: Even faster algorithms for {CSAT} over
  supernilpotent algebras.
\newblock In: 45th International Symposium on Mathematical Foundations of
  Computer Science (MFCS 2020), \emph{Leibniz International Proceedings in
  Informatics (LIPIcs)}, vol. 170, pp. 55:1--55:13. Schloss
  Dagstuhl--Leibniz-Zentrum f{\"u}r Informatik, Dagstuhl, Germany (2020).
\newblock \doi{10.4230/LIPIcs.MFCS.2020.55}

\bibitem{Kompatscher19}
Kompatscher, M.: Notes on extended equation solvability and identity checking
  for groups.
\newblock Acta Math. Hungar. \textbf{159}(1), 246--256 (2019).
\newblock \doi{10.1007/s10474-019-00924-7}.
\newblock \urlprefix\url{https://doi.org/10.1007/s10474-019-00924-7}

\bibitem{LohSen06}
Lohrey, M., S{\'e}nizergues, G.: Theories of {HNN}-extensions and amalgamated
  products.
\newblock In: ICALP 2006, Proceedings, pp. 504--515 (2006).
\newblock \doi{10.1007/11787006_43}

\bibitem{LohreyW19}
Lohrey, M., Wei{\ss}, A.: The power word problem.
\newblock In: 44th International Symposium on Mathematical Foundations of
  Computer Science, {MFCS} 2019, Proceedings, \emph{LIPIcs}, vol. 138, pp.
  43:1--43:15. Schloss Dagstuhl -- Leibniz-Zentrum f{\"{u}}r Informatik (2019).
\newblock \urlprefix\url{http://www.dagstuhl.de/dagpub/978-3-95977-117-7}

\bibitem{mak77}
Makanin, G.S.: The problem of solvability of equations in a free semigroup.
\newblock Math. Sbornik \textbf{103}, 147--236 (1977).
\newblock English transl. in Math. USSR Sbornik 32 (1977)

\bibitem{Matijasevic}
Matijasevic, Y.V.: Enumerable sets are diophantine.
\newblock Soviet Math. Dokl. \textbf{11}, 354--358 (1970).
\newblock \urlprefix\url{https://ci.nii.ac.jp/naid/10009422455/en/}

\bibitem{pap94}
Papadimitriou, C.H.: Computational Complexity.
\newblock Addison Wesley (1994)

\bibitem{Robinson96book}
Robinson, D.J.S.: A course in the theory of groups, \emph{Graduate Texts in
  Mathematics}, vol.~80, second edn.
\newblock Springer-Verlag, New York (1996).
\newblock \doi{10.1007/978-1-4419-8594-1}.
\newblock \urlprefix\url{https://doi.org/10.1007/978-1-4419-8594-1}

\bibitem{Romankov79}
Roman'kov, V.: Equations in free metabelian groups.
\newblock Siberian Mathematical Journal \textbf{20} (1979).
\newblock \doi{10.1007/BF00969959}

\bibitem{Weiss20}
Wei{\ss}, A.: {Hardness of Equations over Finite Solvable Groups Under the
  Exponential Time Hypothesis}.
\newblock In: 47th International Colloquium on Automata, Languages, and
  Programming (ICALP 2020), \emph{Leibniz International Proceedings in
  Informatics (LIPIcs)}, vol. 168, pp. 102:1--102:19. Schloss
  Dagstuhl--Leibniz-Zentrum f{\"u}r Informatik, Dagstuhl, Germany (2020).
\newblock \doi{10.4230/LIPIcs.ICALP.2020.102}

\end{thebibliography}

\end{document}